\documentclass[runningheads]{llncs}
\usepackage{comment}
\usepackage[utf8]{inputenc}
\usepackage{amsmath}

\usepackage{amsthm}
\usepackage{amssymb}
\usepackage{amsfonts}
\usepackage{color,soul}
\usepackage{color,colortbl}
\usepackage{footmisc}
\usepackage{hyperref}
\usepackage{enumitem}

\hypersetup{
    colorlinks=false,
    linkcolor=red,
    filecolor=red,      
    urlcolor=red,
    citecolor=blue,
}

\usepackage[table,dvipsnames]{xcolor}
\usepackage{graphicx}
\usepackage{longtable}
\usepackage{subcaption}
\usepackage[compatibility=false,labelfont=bf,labelsep=period]{caption}
\newcommand{\incomplete}[0]{\hl{(INCOMPLETE)}}
\newcommand{\verify}[0]{\hl{(VERIFY)}}

\definecolor{gray}{gray}{0.9}

\newcommand*\justify{%
  \fontdimen2\font=0.4em
  \fontdimen3\font=0.2em
  \fontdimen4\font=0.1em
  \fontdimen7\font=0.1em
  \hyphenchar\font=`\-
}

\newcommand{\set}[1]{\{#1\}}
\newcommand{\setsize}[1]{\lvert#1\rvert}

\usepackage{stmaryrd}

\mathchardef\mathhyphen="2D
\newcommand{\partition}{\mathcal{P}}
\newcommand{\intercomm}{{\mathit{IC}}}

\newif\iflong
\longtrue
\newcommand{\inLongVersion}[1]{\iflong #1\fi}
\newcommand{\inShortVersion}[1]{\iflong  \else #1\fi}

\begin{document}

\title{On the Hierarchical Community Structure of Practical Boolean Formulas}
\titlerunning{On the HCS of Practical Boolean Formulas}


\author{
Chunxiao Li\inst{1\thanks{Joint first author}\thanks{Work done in part while the authors were at the 2021 Satisfiability: Theory, Practice, and Beyond program at the Simons Institute, Berkeley, CA, USA.}}
\and Jonathan Chung\inst{1\footnotemark[1]}
\and Soham Mukherjee\inst{1, 2}
\and Marc Vinyals\inst{3}
\and \\ Noah Fleming\inst{4}
\and Antonina Kolokolova\inst{5}
\and Alice Mu\inst{1}
\and Vijay Ganesh\inst{1}
}
\authorrunning{Li, Chung et al.}

\institute{
University of Waterloo, Waterloo, Canada
\and Perimeter Institute for Theoretical Physics, Waterloo, Canada
\and Technion, Haifa, Israel
\and University of Toronto, Toronto, Canada
\and Memorial University of Newfoundland, St. John's, Canada
}

\maketitle
\setcounter{footnote}{0} 


\vspace{-0.45cm}
\begin{abstract}
Modern CDCL SAT solvers easily solve industrial instances containing tens of millions of variables and clauses, despite the theoretical intractability of the SAT problem. This gap between practice and theory is a central problem in solver research. It is believed that SAT solvers exploit structure inherent in industrial instances, and hence there have been numerous attempts over the last 25 years at characterizing this structure via parameters. These can be classified as \textit{rigorous}, i.e., they serve as a basis for complexity-theoretic upper bounds (e.g., backdoors), or \textit{correlative}, i.e., they correlate well with solver run time and are observed in industrial instances (e.g., community structure). Unfortunately, no parameter proposed to date has been shown to be both strongly correlative and rigorous over a large fraction of industrial instances.

Given the sheer difficulty of the problem, we aim for an intermediate goal of proposing a set of parameters that is strongly correlative and has good theoretical properties. Specifically, we propose parameters based on a graph partitioning called Hierarchical Community Structure (HCS), which captures the recursive community structure of a graph of a Boolean formula. We show that HCS parameters are strongly correlative with solver run time using an Empirical Hardness Model, and further build a classifier based on HCS parameters that distinguishes between easy industrial and hard random/crafted instances with very high accuracy. We further strengthen our hypotheses via scaling studies. On the theoretical side, we show that counterexamples which plagued flat community structure do not apply to HCS, and that there is a subset of HCS parameters such that restricting them limits the size of embeddable expanders. 
\end{abstract}
\section{Introduction}
\label{sec:introduction}
Over the last two decades, Conflict-Driven Clause-Learning (CDCL) SAT solvers have had a dramatic impact on many sub-fields of software engineering~\cite{cadar2008exe}, formal methods~\cite{clarke2018model}, security~\cite{dolby2007security,xie2005security}, and AI~\cite{blum1997fast}, thanks to their ability to solve large real-world instances with tens of millions of variables and clauses~\cite{SATcomp}, notwithstanding the fact that the Boolean satisfiability (SAT) problem is known to be NP-complete and is believed to be intractable~\cite{cook1971complexity}. 
A plausible explanation of this apparent contradiction would be that NP-completeness
of the SAT problem is established in a worst-case setting, while the dramatic efficiency of modern SAT solvers is witnessed over ``practical'' instances. However, despite over two decades of effort, we still do not have an appropriate mathematical characterization of practical instances (or a suitable subset thereof) and attendant complexity-theoretic upper and lower bounds. 
This gap between theory and practice is rightly considered one of the central problems in solver research by theorists and practitioners alike.


The fundamental premise in this line of work is that SAT solvers are able to find short proofs (if such proofs exist) in polynomial time (i.e., they are efficient) for industrial instances and that they are able to do so because they somehow exploit the underlying properties (a.k.a. structure) of such industrial Boolean formulas\footnote{The term industrial is loosely defined to encompass instances obtained from hardware and software testing, analysis, and verification applications.}, and, further, that hard randomly-generated or crafted instances are difficult because they do not possess such structure.
Consequently, considerable work has been done in characterizing the structure of industrial instances via parameters. The parameters discussed in literature so far can be broadly classified into two categories: correlative and rigorous\footnote{Using terminology by Stefan Szeider~\cite{szeider21talk}.}. The term \emph{correlative} refers to parameters that take a specific range of values in industrial instances (as opposed to random/crafted) and further have been shown to correlate well with solver run time. This suggests that the structure captured by such parameters might explain why solvers are efficient. An example of such a parameter is modularity (more generally community structure~\cite{ansotegui2012community}). By contrast, the term \emph{rigorous} refers to parameters that characterize classes of formulas that are fixed-parameter tractable (FPT), such as backdoors~\cite{williams2003backdoors,zulkoski2018learning}, backbones~\cite{monasson1999determining}, treewidth, and branchwidth~\cite{alekhnovich2011branchwidth,samer2009fixed}, among many others~\cite{samer2009fixed}, or have been used to prove complexity-theoretic bounds over randomly-generated classes of formulas such as clause-variable ratio (a.k.a., density)~\cite{coarfa2003,selman1996generating}.

The eventual goal in this context is to discover a parameter or set of parameters that is both strongly correlative and rigorous, such that it can then be used to establish parameterized complexity-theoretic bounds on an appropriate mathematical abstraction of CDCL SAT solvers, thus finally settling this decades-long open question. Unfortunately, the problem with all the previously proposed rigorous parameters is that either ``good'' ranges of values for these parameters are not witnessed in industrial instances (e.g., such instances can have both large and small backdoors) or they do not correlate well with solver run time (e.g., many industrial instances have large treewidth and yet are easy to solve, and treewidth alone does not correlate well with solving time~\cite{mateescu2011treewidth}).

Consequently, many attempts have been made at discovering correlative parameters that could form the basis of rigorous analysis~\cite{ansotegui2012community,giraldez2015modularity}. Unfortunately, all such correlative parameters either seem to be difficult to work with theoretically (e.g., fractal dimension~\cite{ansotegui2014fractal}) or have obvious counterexamples, i.e., it is easy to show the existence of formulas that simultaneously have ``good'' parameter values and are provably hard-to-solve. For example, it was shown that industrial instances have high modularity, i.e., supposedly good community structure~\cite{ansotegui2012community}, and that there is good-to-strong correlation between modularity and solver run time~\cite{newsham2014community}. However, Mull et al.~\cite{mull-sat16} later exhibited a
family of formulas that have high modularity and require exponential-sized proofs to refute. Finally, this line of research suffers from important methodological issues, that is, experimental methods and evidence provided for correlative parameters tend not to be consistent across different papers in the literature.

\noindent{\bf Hierarchical Community Structure of Boolean Formulas:} Given the sheer difficulty of the problem, we aim for an intermediate goal of proposing a set of parameters that is strongly correlative and has good theoretical properties. Specifically, we propose a set of parameters based on a graph-theoretic structure called Hierarchical Community Structure (HCS), inspired by a commonly-studied concept in the context of hierarchical networks~\cite{clauset2008hierarchy,ravasz2002hierarchical}, which satisfies all the empirical tests hinted above and has better theoretical properties than previously proposed correlative parameters. The intuition behind HCS is that it neatly captures the structure present in human-developed systems which tend to be modular and hierarchical~\cite{simon1962architecture}, and we expect this structure to be inherited by Boolean formulas modelling these systems.

\vspace{0.05cm}
\noindent{\textbf{Contributions\footnote{Instance generator and data can be found at \url{https://satsolvercomplexity.github.io/hcs}. \inShortVersion{Also, for the full-length paper and appendices (with proofs of theorems in Section~\ref{sec:theoretical-results}), please refer to the arXiv version of the paper~\cite{li2021hierarchical}.}}:} 
}
\vspace{-0.2cm}
\begin{enumerate}
\item \textbf{Empirical Result 1 (HCS and Industrial Instances):} We show that a set of parameters based on the HCS of the variable-incidence graph (VIG) of Boolean formulas are effective in distinguishing industrial instances from random/crafted ones. Moreover, we build a classifier that robustly classifies SAT instances into the categories they belong to (verification, random, etc.).
  The classification accuracy is approximately 99\% and we perform a variety of tests to ensure there is no overfitting (See Section~\ref{sec:empirical-result-category}).
    
    \item \textbf{Empirical Result 2 (Correlation between HCS and Solver Run Time):} We build an empirical hardness model based on our HCS parameters to predict the solver run time for a given problem instance. Our model, based on regression, performs well, achieving an $R^2$ score of $0.83$, much stronger than previous such results (See Section~\ref{sec:empirical-result-hardness})
    
    \item \textbf{Empirical Result 3 (Scaling Experiments of HCS Instances):} We empirically show, via scaling experiments, that HCS parameters such as community degree and leaf-community size positively correlate with solving time. We empirically demonstrate that formulas whose HCS decompositions fall in a good range of parameter values are easier to solve than instances with a bad range of HCS parameter values (See Section~\ref{sec:scaling-experiments-HCS-generator}).
    
    \item \textbf{Theoretical Results:} We theoretically justify our choice of HCS by showing that it behaves better than other parameters. More concretely, we show the advantages of hierarchical over flat community structure by identifying HCS parameters which let us avoid hard formulas that can be used as counterexamples to community structure~\cite{mull-sat16}, 
    and by showing graphs where HCS can find the proper communities where flat modularity cannot. We also show that there is a subset of HCS parameters (leaf-community size, community degree, and fraction of inter-community edges) such that restricting them limits the size of embeddable expanders (See Section~\ref{sec:theoretical-results}).
    
    \item \textbf{Instance Generator:} Finally, we provide an HCS-based instance generator which takes input values of our proposed parameters and outputs a formula that satisfies those values. This generator can be used to generate ``easy" and ``hard" formulas with different hierarchical structures (See Section~\ref{sec:scaling-experiments-HCS-generator}).
\end{enumerate}

\noindent{\bf Research Methodology:} We also codify a set of empirical tests which we believe parameters must pass in order to be considered for further theoretical analysis. While other researchers have considered one or more of these tests, we bring them together into a coherent and sound research methodology that can be used for future research in formula parameterization (See Section~\ref{sec:research-methodology}).
We believe that the combination of these tests provides a strong basis for a correlative parameter to be considered worthy of further analysis.

\section{Preliminaries}
\label{sec:preliminaries}

\noindent{\bf Variable Incidence Graph (VIG):} Researchers have proposed a variety of graphs to study graph-theoretic properties of Boolean formulas. In this work we focus on the Variable Incidence Graph (VIG), primarily due to the relative ease of computing community structure over VIGs compared to other graph representations. The VIG for a formula $F$ over variables $x_1,\ldots, x_n$ has $n$ vertices, one for each variable. There is an edge between vertices $x_i$ and $x_j$ if both $x_i$ and $x_j$ occur in some clause $C_k$ in $F$.
One drawback of VIGs is that a clause of width $w$ corresponds to a clique of size $w$ in the VIG. Therefore, large width clauses (of size $n^\varepsilon$) can significantly distort the structure of a VIG, and formulas with such large width clauses should have their width reduced (via standard techniques) before using a VIG.

\noindent{\bf Community Structure and Modularity:} Intuitively, a set of variables (vertices in the VIG) of a formula forms a community if these variables are more densely connected to each other than to variables outside of the set. An (optimal) community structure of a graph is a partition $P=\set{V_1,\ldots,V_k}$ of its vertices into communities that optimizes some measure capturing this intuition, for instance modularity~\cite{newman2004modularity}, which is the one we use in this paper. Let $G=(V,E)$ be a graph with adjacency matrix $A$ and for each vertex $v \in V$ denote by $d(v)$ its degree. Let $\delta_P\colon V \times V \rightarrow \{0,1\}$ be the community indicator function of a partition, i.e.\ $\delta_P(u,v) = 1$ iff vertices $u$ and $v$ belong to the same community in $P$. The \emph{modularity} of the partition $P$ is
\begin{equation}
 Q(P):=\frac{1}{2|E|} \sum_{u,v \in V}\left[A_{u,v} - \frac{d(u) d(v)}{2|E|}\right]\delta_P(u,v)
\end{equation}
Note that $Q(P)$ ranges from $-0.5$ to $1$, with values close to $1$
indicating good community structure. We define the modularity $Q(G)$ of a graph
$G$ as the maximum modularity over all possible partitions, with corresponding partition $\mathcal{P}(G)$.
Other measures may produce radically different partitions.

\noindent{\bf Expansion of a Graph:} Expansion is a measure of graph connectivity~\cite{expanders-survey}. Out of several equivalent such measures, the most convenient to relate to HCS is \emph{edge expansion}: given a subset of vertices $S\subseteq V$, its edge expansion is $h(S)=\setsize{E(S,V \backslash S)} / \setsize{S}$, and the edge expansion of a graph is $h(G) = \min_{1 \leq |S| \leq n/2} h(S)$. A graph family $G_n$ is an expander if $h(G_n)$ is bounded away from zero. Resolution lower bounds (of both random and crafted formulas) often rely on strong expansion properties of the graph \cite{ben2001short}.

\section{Research Methodology}
\label{sec:research-methodology}
As stated above, the eventual goal of the research presented here is to discover a structure and an associated parameterization that is highly correlative with solver run time, is witnessed in industrial instances, and is rigorous, i.e., forms the basis for an upper bound on the parameterized complexity~\cite{samer2009fixed} of the CDCL algorithm. Considerable work has already been done in attempting to identify exactly such a set of parameters~\cite{newsham2014community}. However, we observed that there is a wide diversity of research methodologies adopted by researchers in the past. We bring together the best lessons learned into what we believe to be a sound, coherent, and comprehensive research methodology explained below. We argue that every set of parameters must meet the following empirical requirements in order to be considered correlative:

\begin{enumerate}[leftmargin=*]
    \item \textbf{Structure of Industrial vs. Random/Crafted Instances:} A requisite for a structure to be considered correlative is that industrial instances must fall within a certain range of values for the associated parameters, while random and crafted instances must have a different range. An example of such a structure is the community structure of the VIG of Boolean formulas, as parameterized by modularity. Multiple experiments have shown that industrial instances have high modularity (close to $1$), while random instances tend to have low modularity (close to $0$)~\cite{newsham2014community}. This could be demonstrated via a correlation experiment or by building a classifier that takes parameter values as input features.
    
    \item \textbf{Correlation between Structure and Solver Run Time:} Another requirement is correlation between parameters of a structure and solver run time. Once again, community structure (and the associated modularity parameter) forms a good example of a structure that passes this essential test. For example, it has been shown that the modularity of the community structure of industrial instances (resp. random instances) correlates well with low (resp. high) solver run time~\cite{newsham2014community}. One may use either correlation methods or suitable machine learning predictors (e.g., random forest) as evidence here.

    \item \textbf{Scaling Studies:} To further strengthen the experimental evidence, we require that the chosen structure and its associated parameters must pass an appropriately designed scaling study. The idea here is to vary one parameter value while keeping as much of the rest of the formula structure constant as possible, and see its effect on solver run time. An example of such a study is the work of Zulkoski et al.~\cite{zulkoski2018effect}, who showed that increasing the mergeability metric has a significant effect on solver run time.
\end{enumerate}

\noindent{\bf Limitations of Empirical Conclusions:} As the reader is well aware, any attempt at empirically discovering a suitable structure (and associated parameterization) of Boolean formulas and experimentally explaining the power of solvers is fraught with peril, since all such experiments involve pragmatic design decisions (e.g., which solver was used, choice of benchmarks, etc.) and hence may lead to contingent or non-generalizable conclusions. For example, one can never quite eliminate a parameter from further theoretical analysis based on empirical tests alone, for the parameter may fail an empirical test on account of benchmarks considered or other contingencies. Another well-understood issue with conclusions based on empirical analysis alone is that they by themselves cannot imply provable statements about asymptotic behavior of algorithms. However, one can use empirical analysis to check or expose gaps between the behavior of an algorithm and the tightness of asymptotic statements (e.g., the gap between efficient typical-case behavior vs. loose worst-case statements). Having said all this, we believe that the above methodology is a bare minimum that a set of parameters must pass before being considered worthy of further theoretical analysis. In Section~\ref{sec:experimental-design}, we go into further detail about how we protect against certain contingent experimental conclusions.

\noindent{\bf Limits of Theoretical Analysis:} Another important aspect to bear in mind is that it is unlikely any small set of parameters can cleanly separate all easy instances from hard ones. At best, our expectation is that we can characterize a large subset of easy real-world instances via the parameters presented here, and thus take a step towards settling the central question of solver research.

\section{Hierarchical Community Structure}
\label{sec:hcs}
Given that many human-developed systems are modular and hierarchical~\cite{simon1962architecture}, it is natural to hypothesize that these properties are transferred over to Boolean formulas that capture the behaviour of such systems. We additionally hypothesize that purely randomly-generated or crafted formulas do not have these properties of hierarchy and modularity, and that this difference partly explains why solvers are efficient for the former and not for the latter class of instances. We formalize this intuition via a graph-theoretic concept called Hierarchical Community Structure (HCS), where communities can be recursively decomposed into smaller sub-communities. Although the notion of HCS has been widely studied~\cite{clauset2008hierarchy,ravasz2002hierarchical}, it has not been considered in the context of Boolean formulas before.


\paragraph{\bf Hierarchical Community Structure Definition:} A \emph{hierarchical decomposition} of a graph $G$ is a recursive partitioning of $G$ into subgraphs, represented as a tree $T$. Each node $v$ in the tree $T$ is labelled with a subgraph of $G$, with the root labelled with $G$ itself. The children of a node corresponding to a (sub)graph $H$ are labelled with a partitioning of $H$ into subgraphs $\{H_1,\ldots, H_k\}$; see Figure~\ref{fig:HCS-decomposition}. There are many ways to build such hierarchical decompositions. The method that we choose constructs the tree by recursively maximizing the modularity, as in the hierarchical multiresolution method~\cite{granell2012hierarchical}. We call this the HCS decomposition of a graph $G$: for a node $v$ in the tree $T$ corresponding to a subgraph $H$ of $G$, we construct $|\mathcal{P}(H)|$ children, one for each of the subgraphs induced by the modularity-maximizing partition $\mathcal{P}(H)$, unless $|\mathcal{P}(H)|=1$, in which case $v$ becomes a \emph{leaf} of the tree. In the case of HCS decompositions, we refer to the subgraphs labelling the nodes in the tree as \emph{communities} of $G$.

\begin{figure}[t]
    \hfil
    \includegraphics[page=1,scale=.8]{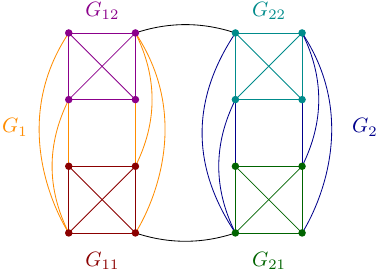}
    \hfil
    \includegraphics[page=2,scale=.8]{figures/hcs-pic.pdf}
    \hfil
    \caption{A hierarchical decomposition (right) constructed by recursively maximizing the modularity of the graph (left).} \label{fig:HCS-decomposition}
\end{figure}

We are interested in comparing the hierarchical community structures of Boolean formulas in conjunctive normal form, represented by their VIGs. For this comparison, we use the following parameters:
\begin{itemize}[leftmargin=*]
    \item The \emph{community degree} of a community in a HCS decomposition is the number of children of its corresponding node.
    \item A \emph{leaf-community} is one with degree $0$.
    \item The \emph{size} of a community is its number of vertices.
    \item The \emph{depth} or \emph{level} of a community is its distance from the root.
    \item The \emph{inter-community edges} of a partition $\mathcal{P}(H)$ are $E_\intercomm(H)=\bigcup_{H_i,H_j \in\mathcal{P}(H)}\allowbreak E(H_i,H_j)$, the edges between all pairs of subgraphs, and their endpoints $V_\intercomm(H)=\bigcup E_\intercomm$ are the \emph{inter-community vertices}. Note that
    $2\setsize{E_\intercomm(H)}/\setsize{H}$ is an upper bound for the edge expansion of $H$. 
\end{itemize}

Note that these parameters are not independent. For example, changes in the number of inter-community vertices or inter-community edges will affect modularity. Since our hierarchical decomposition is constructed using modularity, this could affect the entire decomposition and hence the other parameters.

\section{Empirical Results}
\label{sec:empirical-results}
    We now turn to the results of our empirical investigations with HCS parameters. We computed $49$ unique parameters capturing the HCS structure, together with several base parameters
    measuring different structural properties of input VIGs\footnote{For a complete list, see: \url{https://satsolvercomplexity.github.io/hcs/data}}. To compute the hierarchical community structure, we used the Louvain method~\cite{blondel2008louvain} to detect communities and recursively call the Louvain method to produce a hierarchical decomposition. The Louvain method is considered to be more efficient and produces higher-modularity partitions than other known algorithms.
    
\label{sec:experimental-design}

\vspace{0.2cm}
\noindent{\bf Experimental Design.} In our experiments we used a set of 10\,869 instances from five classes, which we believe is sufficiently large and diverse to draw sound empirical conclusions \inLongVersion{(See Table~\ref{tab:benchmark_instances} in Appendix~\ref{app:benchmark_instances})}\inShortVersion{(See Appendix~\cite{li2021hierarchical})}. We did not explicitly balance the ratio of satisfiable instances in our benchmark selection because we expect our methods to be sufficiently robust as long as the benchmark contains a sufficient number of SAT and UNSAT instances.

In order to get interesting instances for modern solvers, we considered formulas which were previously used in the SAT competition from 2016 to 2018~\cite{SATcomp}. Specifically, we took instances from five major tracks of the competition: agile, verification, crypto, crafted, and random. We also generated additional instances for some classes: for verification, we scaled the number of unrolls when encoding finite state machines for bounded model checking; for crypto, we encoded SHA-1 and SHA-256 preimage problems; for crafted, we generated combinatorial problems using \texttt{cnfgen}~\cite{cnfgen2017}; and for random, we generated $k$-CNFs at the corresponding threshold CVRs for \(k \in \{3, 5\}\), again using \texttt{cnfgen}. A summary of the instances is presented in the Appendix.

We preprocessed all formulas using the MiniSAT preprocessor~\cite{een2005preprocessing}, and used MapleSAT~\cite{liang2016maplesat} as our CDCL solver of choice since it is a leading and representative solver. The core of the preprocessing 
was a combination of 
variable elimination with subsumption and self-subsuming resolution~\cite{een2005preprocessing}. For computing satisfiability and running time, we used SHARCNET's Intel E5-2683 v4 (Broadwell) 2.1 GHz processors~\cite{sharcnet-graham}, limiting the computation time to 5\,000 seconds\footnote{This value is the time limit used by the SAT competition.}. For parameter computation
we did not limit the type of processor because structural parameter values are independent of processing power.



\subsection{\bf HCS-based Category Classification of Boolean Formulas}
\label{sec:empirical-result-category}
The question whether our set of HCS parameters is able to capture the underlying structure that differentiates industrial instances from the rest naturally lends itself to a classification problem. Therefore, we built a multi-class Random Forest classifier to classify a given SAT instance into one of the five categories: verification, agile, random, crafted, or crypto. Random Forests~\cite{Breiman2001RandomForest} 
can
learn complex, highly non-linear relationships while having simple structure, and hence are easier to interpret than other models (e.g., deep neural networks).

We used an off-the-shelf implementation of a Random Forest classifier implemented as \texttt{sklearn.ensemble.RandomForestClassifier} in scikit-learn~\cite{scikit-learn}. Using the default set of parameters in scikit-learn version 0.24, we trained our classifier using 800 randomly sampled instances of each category on a set of 49 features to predict the class of the problem instance. We found that our classifier performs extremely well, giving an average accuracy score of \(0.99\) over 5 cross-validation datasets. Further, the accuracy did not depend on our choice of classifier. In particular, we found similar accuracy scores when we used C-Support Vector classification~\cite{Platt99probabilisticoutputs} instead of Random Forests. 

We also determined the five most important features used by our classifier. 
Since several features in our feature set are highly correlated, we first performed a hierarchical clustering on the feature set based on Spearman rank-order correlations. From the 22 clusters that were generated, we arbitrarily chose a single feature from each cluster as a representative member of the cluster~f\footnote{See \url{https://satsolvercomplexity.github.io/hcs/data} for details on clusters.}. Using these 22 representative features, we then computed their importance using permutation importance~\cite{Breiman2001RandomForest}. In Table~\ref{tab:randomforest} we list the top five representative features from each cluster, not necessarily in order of importance.

\begin{table}[t]
\bgroup
\def\arraystretch{1.2}%
\begin{center}
\caption{Results for classification and regression experiments with HCS parameters. For regression we report $R^2$ values, whereas for classification we report the mean of the balanced accuracy score over 5 cross-validation datasets.}
\begin{tabular}{ m{0.19\textwidth} | m{0.39\textwidth}| m{0.26\textwidth}  }
 & \textbf{Category} & \textbf{Runtime} \\
\hline
Score & \(0.996\pm 0.001\) & \(0.825 \pm 0.016\) \\
\hline
Top 5 features & \texttt{rootMergeability} \newline \texttt{maxInterEdges/CommunitySize} \newline
\texttt{cvr} \newline \texttt{leafCommunitySize} \newline \texttt{lvl2InterEdges/lvl2InterVars} & \texttt{rootInterEdges} \newline \texttt{lvl2Mergeability} \newline \texttt{cvr} \newline
\texttt{leafCommunitySize} \newline \texttt{lvl3Modularity} 
\end{tabular}
\label{tab:randomforest}
\end{center}
\egroup
\end{table}

\subsection{\bf HCS-based Empirical Hardness Model}
\label{sec:empirical-result-hardness}
We used our HCS parameters to build an empirical hardness model (EHM) to predict the run time of MapleSAT on a given instance. Since the solving time is a continuous variable, we considered a regression model built using Random Forests, namely 
\texttt{\justify{sklearn.ensemble.RandomForestRegressor}} from
scikit-learn~\cite{scikit-learn}. Before training our regression model, we removed instances which timed-out at 5\,000 seconds and those instances that were solved almost immediately (in zero seconds) to avoid issues with artificial cut-off boundaries. We then trained our Random Forest model using the default set of parameters in scikit-learn version 0.24 to predict the logarithm of the solving time using the remaining 1\,880 instances, equally distributed between different categories.

We observed that our regression model performs quite well, with an  $R^2$ score~\cite{Steel60} of 0.83, which implies that in the training set, almost 83\% of the variability of the dependent variable (i.e., in our case, the logarithm of the solving time) is accounted for, and the remaining 17\% is still unaccounted for by our choice of parameters. Similar to category classification, we also looked for the top five predictive features used by our Random Forest regression model using the exact same process. We list the representative features in Table~\ref{tab:randomforest}.

Additionally, we trained our EHM on each category of instances separately. We found that the performance of our EHM varies with instance category. Concretely, agile outperformed all other categories with an average $R^2$ value of $0.94$, followed by random, crafted and verification instances with scores of $0.81, 0.85$ and $0.74$ respectively. The worst performance was shown by the instances in crypto, with a score of $0.48$. 

\subsection{HCS Parameter Value Ranges for Industrial/Random Instances}\label{sec:industrial-params} 

In the previous section, we reported on the top five parameters most predictive of the solver runtime in the context of our Random Forest regression model. These parameters can be divided into five distinct classes of parameters: mergeability-based, modularity-based, inter-community edge based, CVR, and leaf-community size. The parameters CVR, mergeability and modularity have been studied by previous work. CVR~\cite{cheeseman1991cvr} is perhaps the most studied parameter among the three. Zulkoski~et al. ~\cite{zulkoski2018effect} showed that mergeability, along with combinations of other parameters, correlates well with solver run time; Ansotegui~et al.~\cite{ansotegui2012community} showed that industrial instances have good modularity compared to random instances; and Newsham et al.~\cite{newsham2014community} showed that modularity has good-to-strong correlation with solver run time. We examined the remaining parameters, i.e. inter-community edge based parameters (\texttt{rootInterEdges})  and leaf-community size to gain a better understanding of the impact of these parameters on the problem structure and solver runtime, respectively. In this subsection, we look at how HCS parameters scale as the size of industrial instances increases. And in Section~\ref{sec:scaling-experiments-HCS-generator}, we introduce a HCS instance generator, which we use to perform a set of controlled experiments.  We then discuss how the hardness of the instances changes when certain HCS parameters are increased/decreased.

\begin{figure}[t]
\centering     
{\includegraphics[width=0.8\linewidth]{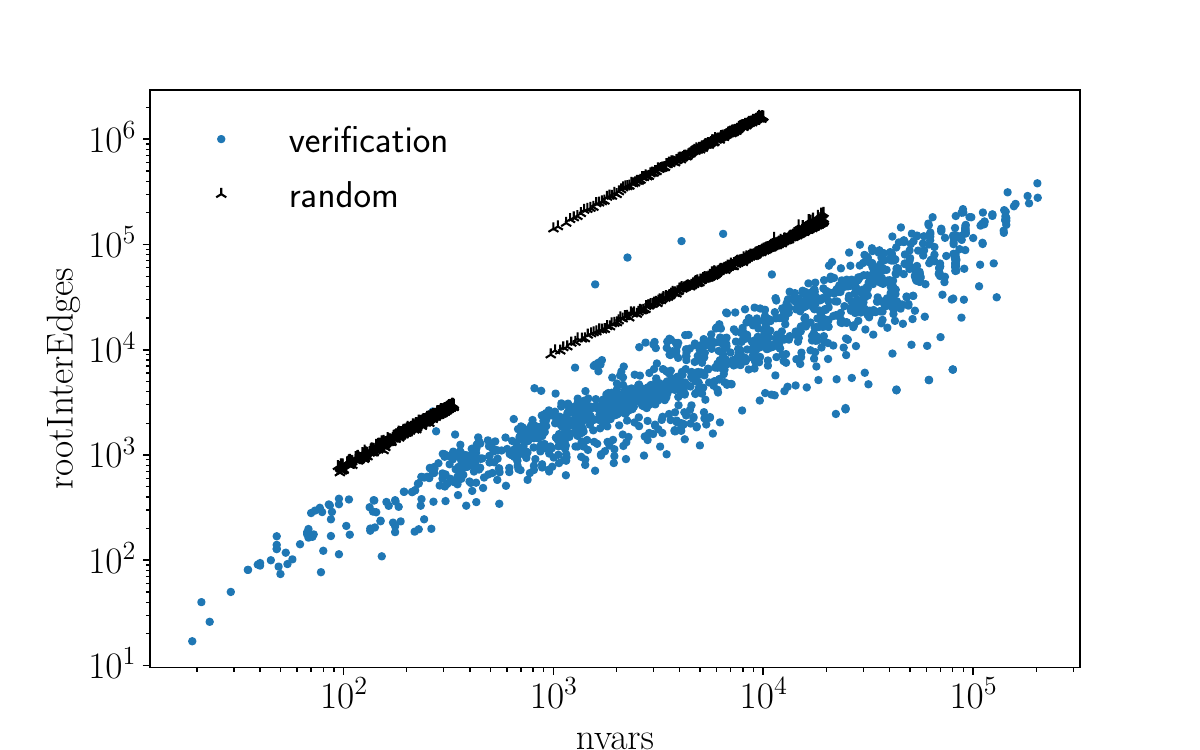}}
\caption{\label{fig:scaling_inter_edges_n} Dependence of the number of inter-community edges  at the root level (rootInterEdges) vs. the number of variables in a formula, for verification and random instances in our dataset. The two distinct lines (starting from the bottom) for random instances correspond to 3-CNFs and 5-CNFs, respectively.
}

\end{figure}

\paragraph{\bf Observations.}
We observe that hierarchical decomposition generally produces leaf communities of maximal size comparable to the largest clause width,  
except for very unbalanced formulas (easy for other reasons). The community degree is highest at root level of every instance, and seems to be bounded by  $O(\log n)$. 
This fits within the range of parameters considered in Section~\ref{sec:theoretical-results}.


In Figure~\ref{fig:scaling_inter_edges_n}, we show how the inter-community edge based parameter \texttt{\justify{rootInterEdges}} scales with the number of variables in a formula, for verification and random instances. 
We note that for random instances,  \texttt{rootInterEdges} grows linearly with the instance size, whereas in verification instances it grows sublinearly.
This supports our intuition that graphs of hard (random) instances are expanders, whereas graphs of industrial instances are not. 

\subsection{Scaling Experiments with HCS parameters}
\label{sec:scaling-experiments-HCS-generator}

\paragraph{\bf Instance Generator.}
To isolate the effects of HCS parameters on solver runtime, we built an HCS instance generator to construct SAT instances with varying leaf-community size and other HCS parameters. On a high level, the instance generator constructs instances bottom-up, starting with random disjoint formulas of predefined CVR as leaf communities, then combining them recursively by introducing bridge clauses with variables in at least two sub-communities to form super-communities at that level, which in turn are combined at the following level. We point out that in our generator, modularity is specified implicitly through the above parameters, and we do not control for mergeability at all. We refer the reader to the works by Zulkoski et al.~\cite{zulkoski2018effect} and Gir\'aldez-Cru~\cite{giraldez2016} for literature on the empirical behaviours of mergeability and power law, respectively.

It is important to note that our HCS instance generator is not intended to be perfectly representative of real-world instances. In fact, there are multiple properties of our generated instances which are not reflective of industrial instances. For example, our generator assumes that all leaf-communities have the same size and depth, which is demonstrably untrue of industrial instances. In some cases, the communities produced by our generator might not be the same as the communities which would be detected using the Louvain method to perform a hierarchical community decomposition. For example, it might be possible to further decompose the generated ``leaf-communities" into smaller communities. Thus, our generator is only intended to demonstrate the effect of varying HCS parameters on solver runtime.

\paragraph{\bf Observations.} We constructed formulas with varying CVR, power law parameter, hierarchical degree, depth, inter-community edge density, inter-community variable density, and clause width. We found evidence which suggests that increasing any of leaf-community size, depth, or community degree, while keeping every other HCS parameter fixed, increases the overall hardness of the generated formula. For example, we found that changing the size of leaf-communities from 15 variables to 20, the solving time changed from 4.96 seconds to upwards of 5000 seconds. Similarly, changing the depth from 4 to 5 resulted in an increase in solving time from 0.03 seconds to over 5000 seconds.

\subsection{Discussion of Empirical Results}
\label{sec:analysis-of-empirical-results}
The goal of our experimental work was to first ascertain whether HCS parameters can distinguish between industrial and random/crafted instances, and whether these parameters show any correlation with CDCL solver runtime. The robustness of our classifier indicates that HCS parameters are indeed representative of the underlying structure of Boolean formulas from different categories. Further, our empirical hardness model confirms that the correlation of HCS parameters with solver run time is strong---much stronger than previously proposed parameters. We also find that our HCS parameters are more effective in capturing the hardness or easiness of formulas from industrial/agile/random/crafted, but not crypto. The crypto class is an outlier. It is not clear from our experiments (nor any previous ones) as to why crypto instances are hard for CDCL solvers. 

We also identified the top five (representative) parameters in terms of their importance in predicting the category (classification) or runtime of an instance (regression). The accuracy for classification and regression with only the top features features dropped to 0.94 and 0.77, respectively, suggesting that only a few parameters are likely to play a role in closing the question on why solvers are efficient for industrial instances. Note that a classification accuracy of 0.99 is likely to suggest that our model is over-fitting. Fortunately, in our case our models are trained over a large set of instances obtained via very different methods (e.g., random over various widths, different kinds of crafted, verification instances from different domains), and therefore, there is sufficient entropy in our data set so that overfitting is unlikely to be a concern for the robustness of our model.

In our investigation of parameters based on inter-community edges and leaf-community size, we found that industrial instances typically have small average leaf-community size, high modularity, and relatively few inter-community edges, while random/crafted have larger average leaf-community size, low modularity, and a very high number of inter-community edges.
This suggests that leaf-community size and the fraction of inter-community edges, as well as  community degree, are important HCS parameters to consider further.


\section{Theoretical Results}
\label{sec:theoretical-results}

In this section,  we show that hierarchical decomposition avoids some of the pitfalls of flat community structure, a promising correlative parameter for explaining easiness of the industrial instances~\cite{newsham2014community}. Community structure was theoretically shown to be insufficient by Mull et al.~\cite{mull-sat16}, where they showed that formulas with good community structure can have random formulas embedded in them either in a community or over the inter-community edges. To avoid embedding a random formula in a community, its size has to be small (relative to the entire graph), and avoiding expanders over inter-community edges requires that there not be too many communities. A way to be able to restrict both is to consider a hierarchical decomposition, limiting both the number of sub-communities (community degree) in each level of the decomposition, as well as the leaf community size thus avoiding the most important issues that flat community structure suffers from.   

Based on our experimental work, we narrow down the most predictive HCS parameters to be leaf-community size, community degree, and the number inter-community edges in each decomposition. 
These parameters also play a role in our theoretical results below. For a formula to have ``good'' HCS, we restrict the parameter ranges as follows: the graph must exhibit $O(\log n)$ leaf-community size  and community degree, and have a small number of inter-community edges in each decomposition of a community. These assumptions are supported by our experimental results \inLongVersion{(see Section~\ref{sec:industrial-params})}\inShortVersion{(See Appendix~\cite{li2021hierarchical})}. We show that these restrictions are necessary in \inLongVersion{Section~\ref{sec:counterexamples}}\inShortVersion{Appendix}, where we also present a significantly simplified proof of the result of Mull et al.~\cite{mull-sat16}. 


\paragraph{\bf Bounding the Size of Expanders in Good HCS Graphs.}
Ideally, we would like to be able to prove an upper bound on proof size or search time which depends on the HCS parameters of a formula. Unfortunately, our current state of understanding does not allow for that. A step towards such a result would be to show that formulas with good HCS (and associated parameter value ranges) are not susceptible to typical methods of proving resolution lower bounds. Currently, all resolution bounds exploit \emph{expansion} properties -- typically \emph{boundary expansion} -- of the CNF formula (or more precisely its bipartite constraint-variable incidence graph (CVIG)). Therefore our goal is to show that formulas with good HCS parameters have poor expansion properties, and also do not have large expanding subgraphs embedded within them. Note that the VIG is related to the CVIG by taking the square of its adjacency matrix, from where it follows that, for formulas with low width, if the VIG is not edge-expanding then the CVIG is not vertex-expanding. Furthermore, again for formulas with low width, vertex expansion is closely related to boundary expansion. Hence we only need to focus on VIG edge expansion.
With this in mind, we state several positive and negative results.


First, we observe \inLongVersion{(in Section \ref{sec:expanders})}\inShortVersion{(see Appendix)} that if the number of inter-community edges at the top level of the decomposition grows sub-linearly with $n$ and at least two sub-communities contain a constant fraction of vertices, then this graph family is not an expander. Unfortunately, we can also show \inLongVersion{(in Section~\ref{sec:counterexamples})}\inShortVersion{(see Appendix)} that graphs with good HCS can simultaneously have sub-graphs that are large expanders, with the worst case being very sparse expanders, capable of ``hiding'' in the hierarchical decomposition by contributing relatively few edges to any cut. To avoid that, we require an explicit bound on the number of inter-community edges, in addition to small community degree and small leaf-community size. This lets us prove the following statement.
 
\begin{theorem}\label{hcs-expansion}
 Let $G=\{G_n\}$ be a family of graphs. Let $f(n) \in \omega(poly(\log n))$, $f(n) \in O(n)$. Assume that $G$ has HCS with the number of inter-community edges $o(f(n))$ for every community $C$ of size at least $\Omega(f(n))$ and depth is bounded by $O(\log n)$. Then $G$ does not contain an expander of size $f(n)$ as a subgraph.
\end{theorem} 

Note that our experiments show that the leaf size and depth in industrial instances are relatively small and the number of inter-community edges grows slowly. From this and the theorem above, we can show that graphs with \emph{very} good HCS properties do not contain linear-sized expanders. 


\paragraph{\bf Lower Bounds Against HCS:} 
We are also able to show several of strong lower bounds on formulas with good HCS \inLongVersion{(see Section~\ref{sec:counterexamples})}\inShortVersion{(see Appendix)}. For a number of combinations of parameters, we show that restricting ourselves to ``good'' ranges of these parameters does not rule out formulas which require superpolynomial size resolution refutations.  Our most striking counterexample essentially shows that if the degree of the VIG is more than a small constant, then it is possible to embed formulas of superpolynomial resolution complexity. In contrast with the previous results on the size of embeddable expanders in instances with good HCS, this result shows how to embed a sparse expander of superlogarithmic size.

\paragraph{\bf Hierarchical vs. Flat Modularity:} 
It is well-known that modularity suffers from a \emph{resolution limit} and cannot detect communities smaller than a certain threshold~\cite{fortunato2007modularity}, and that HCS can avoid this problem in some instances~\cite{blondel2008louvain}. In \inLongVersion{Section~\ref{sec:resolution_limit}}\inShortVersion{Appendix} we provide an asymptotic, rigorous statement of this observation.

\begin{theorem}
  \label{th:hcs-resolution}
  There exists a graph $G$ whose natural communities are of size $\log(n)$ and correspond to the (leaf) HCS communities, while the partition maximizing modularity consists of communities of size $\Theta\Bigl(\sqrt{n/\log^3n}\Bigr)$.
\end{theorem}

\section{Related Work}
\label{sec:related-work}
\noindent\textbf{Community Structure:} Using modularity to measure community structure allows one to distinguish industrial instances from randomly-generated ones~\cite{ansotegui2012community}. Unfortunately, it has been shown that expanders can be embedded within formulas with high modularity~\cite{mull-sat16}, i.e., there exist formulas that have good community structure and yet are hard for solvers. 

\noindent\textbf{Heterogeneity:} Unlike uniformly-random formulas, the variable degrees in industrial formulas follow a powerlaw distribution~\cite{ansotegui2009powerlaw}. However, degree heterogeneity alone fails to explain the hardness of SAT instances. Some heterogeneous random k-SAT instances were shown to have superpolynomial resolution size \cite{blasius2020heterogeneity}, making them intractable for current solvers.

\noindent\textbf{SATzilla:} SATzilla uses 138 disparate parameters~\cite{xu2012satzilla}, some of which are probes aimed at capturing a SAT solver's state at runtime, to predict solver running time. Unfortunately, there is little or no evidence that most of these parameters are amenable to theoretical analysis.

\noindent\textbf{Clause-Variable Ratio (CVR):} Cheeseman et al.~\cite{cheeseman1991cvr} observed the satisfiability threshold behavior for random k-SAT formulas, where they show formulas are harder when their CVR are closer to the satisfiability threshold. Outside of extreme cases, CVR alone seems to be insufficient to explain hardness (or easiness) of instances, as it is possible to generate both easy and hard formulas with the same CVR~\cite{friedrich2017phasetransitions}. Satisfiability thresholds are poorly defined for industrial instances, and Coarfa et al.~\cite{coarfa2003} demonstrated the existence of instances for which the satisfiability threshold is not equal to the hardness threshold.


\noindent\textbf{Treewidth:} Although there are polynomial-time non-CDCL algorithms for SAT instances with bounded treewidth~\cite{alekhnovich2011branchwidth}, treewidth by itself does not appear to be a predictive parameter of CDCL solver runtime. For example, Mateescu~\cite{mateescu2011treewidth} showed that some easy instances have large treewidth, and later it was shown that treewidth alone does not seem to correlate well with solving time~\cite{zulkoski2018effect}.

\noindent\textbf{Backdoors:}
In theory, the existence of small backdoors~\cite{williams2003backdoors,samer2008backdoor} should allow CDCL solvers to solve instances quickly, but empirically backdoors have been shown not to strongly correlate with CDCL solver run time \cite{kilby2005backdoors}. 


\section{Conclusions and Future Work}
\label{sec:conclusion}
In this paper, we propose HCS as a correlative set of parameters for explaining the power of CDCL SAT solvers over industrial instances, which also has good theoretical properties. Empirically, HCS parameters are much more predictive than previously proposed correlative parameters in terms of classifying instances into random/crafted vs. industrial, and in terms of predicting solver run time. Among the top five most predictive parameters, three are HCS parameters, namely leaf-community size, modularity and fraction of inter-community edges. The remaining two are cvr and mergeability. We further identify the following core HCS parameters that are the most predictive among all HCS parameters, namely, leaf-community size, modularity,
and fraction of inter-community edges. Indeed, these same parameters also play a role in our subsequent theoretical analysis, where we show that counterexamples to flat community structure do not apply to HCS, and that restricting certain HCS parameters limits the size of embeddable expanders. In the final analysis, we believe that HCS, along with other parameters such as mergeability or heterogeneity, will play a role in finally settling the question of why solvers are efficient over industrial instances.

\bibliographystyle{splncs04}
\bibliography{reference}

\inLongVersion{
\newpage
\section*{Appendix}
\label{sec:appendix}
\appendix

\section{Summary of Benchmark Instances}
\label{app:benchmark_instances}
Table~\ref{tab:benchmark_instances} contains the number of SAT and UNSAT instances per category of instances that were used in the experiments presented in this paper. 
\setlength{\tabcolsep}{8pt}
\begin{table}[t]
    \centering
    \caption{Summary of benchmark instances}
    \begin{tabular}{l| c | c | c }
    \hiderowcolors
        \textbf{Class Name} & \textbf{\#SAT} & \textbf{\#UNSAT} & \textbf{\#UNKNOWN} \\
        \hline
        \texttt{agile} & 372 & 475 & 8 \\
        \texttt{crafted} & 300 & 99 & 617 \\
        \texttt{crypto} & 1052 & 355 & 3496 \\
        \texttt{random} & 276 & 224 & 661 \\
        \texttt{verification} & 197 & 2345 & 392 \\
    \end{tabular}
    \label{tab:benchmark_instances}
\end{table}

\section{Inter-community Edges vs. Number of Variables}
\label{app:scaling_inter_edges_n_all}
In our experiments, we compute the number of inter-community edges in the VIG of each formula and observe how it scales as the number of variables in the formula increases. We note that the behaviour is different for different classes of instances. Figure~\ref{fig:scaling_inter_edges_n_all} contains logarithmically-scaled plots of this behaviour for each of the instance classes considered in this work.

\begin{figure}[t]
\centering     
{\includegraphics[width=\linewidth]{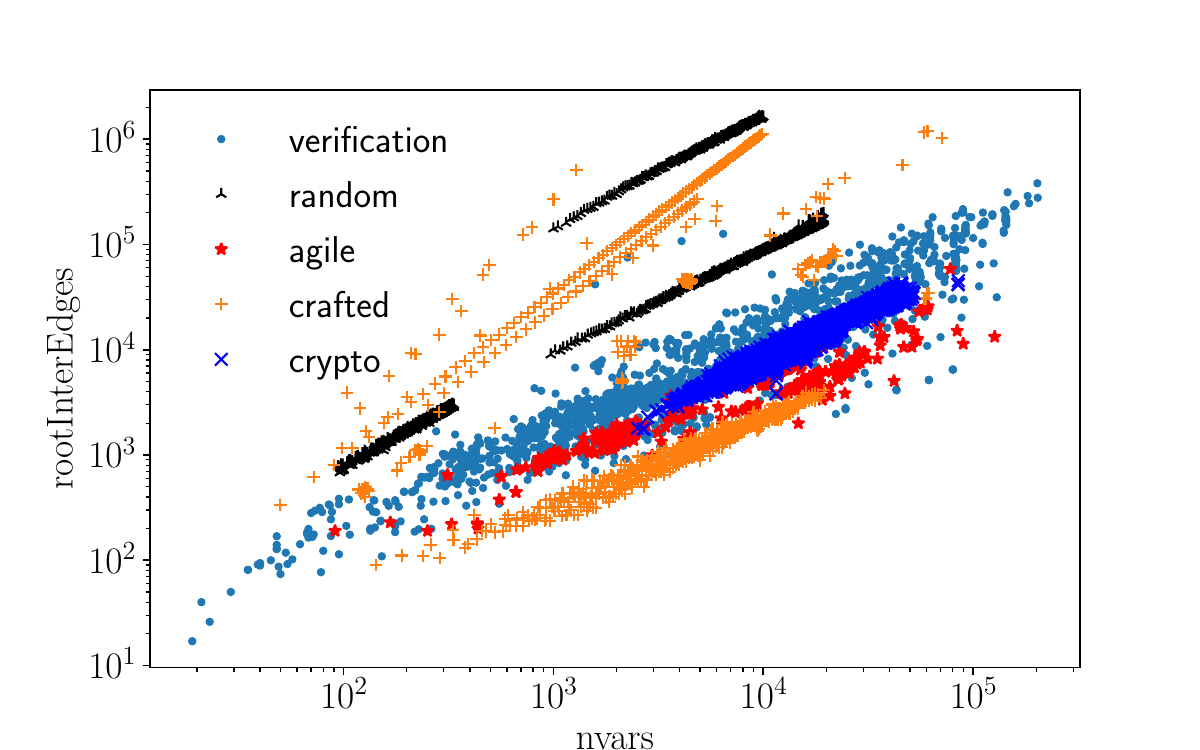}}
\caption{Number of inter-community edges  at the root level (rootInterEdges) vs. the number of variables (nvars) in a formula for all categories.}
\label{fig:scaling_inter_edges_n_all}
\end{figure}


\newpage
\section{Theoretical Results}
\label{sec:detailed-theoretical-results}

\subsection{Embedding expanders in HCS graphs} 
\label{sec:expanders}

We show that a good enough HCS decomposition of a graph disallows large expanders to be embedded in it. Consider graphs with following HCS restrictions.


\begin{enumerate} 
 \item Small leaf size: all leaves are of size at most $O(\log n)$.  
  \item Bounded community degree: for every community at every level, the number of immediate sub-communities is bounded by $O(\log n)$. 
 \item Small number of inter-community edges (the smaller the better).  
\end{enumerate} 

The reason for the small leaf size condition is to avoid the counterexample of Mull, Fremont, and Seshia~\cite{mull-sat16}.
The second condition, bounded community degree, avoids counterexamples in section~\ref{sec:counterexamples}. The last property is needed to avoid the counterexample at the end of this section.

In the following, we use $G$ to refer both to a family of graphs as well as a specific representative of that family (the reason for the family is to be able to talk about asymptotic behaviour of parameters). 

\begin{claim}\label{hcs-flat}  
Suppose a graph $G$ with $|G|=n$ has a community decomposition into $C_1, \dots, C_\ell $ maximizing the modularity (flat, single-level decomposition). Let $H$ be a subgraph of $G$, with 
$|H| = f(n)$ for some $f(n) \in O(n)$. Then either  $H$ is not an expander,  or the number of inter-community edges is $\Omega(f(n))$, or there exists a community containing a subgraph of $H$ of size at least $f(n)(1-o(1))$. 
\end{claim}
\begin{proof} 
Suppose that $H$ is an expander graph of size $f(n)$ with edge expansion $h(H)=\alpha$,  and let $S$ be the largest subset of $H$ which is within a single $C_i$ (wlog, say $S \subseteq C_1$). Suppose that $|S| = \delta \cdot f(n)$ for some constant $\delta$. Then the number of edges out of $C_1$ is at least $\alpha\cdot \min\{\delta,1-\delta\} f(n)$, and thus the the total number of inter-community edges is $\Omega(f(n))$. If $|S| = o(f(n))$,  then it is possible to split communities $C_1 \dots C_t$ into 2 sets with roughly the same number of vertices of $H$ on each side (up to a subconstant factor); the number of edges of $H$ going across this cut should be close to $\alpha f(n)$, again $\Omega(f(n))$.
\end{proof} 


\begin{theorem}\label{hcs-embed} 
 Let $G=\{G_n\}$ be a family of graphs. Let $f(n) \in \omega(poly(\log n))$, $f(n) \in O(n)$. Assume that $G$ has HCS with the number of inter-community edges $o(f(n))$ for every community $C$ of size at least $\Omega(f(n))$ and depth is bounded by $O(\log n)$. Then $G$ does not contain an expander of size $f(n)$ as a subgraph.  
 \end{theorem}
\begin{proof} 
Let $H$ be a subgraph of $G$ of size $f(n)$.  Let $C$ be the smallest sub-community in the hierarchical decomposition of $G$ containing at least  $(1-o(1))$-fraction of $H$. Note that if $H$ is an expander, then any subgraph of $H$ on $1-o(1) $ vertices is also an expander. 

First, since $|H| \in \omega(poly(\log n))$, $C$ cannot be a leaf. Therefore, $C$ it will be partitioned into sub-communities $C_1 \dots C_\ell$ by the hierarchical decomposition.  Since we assumed that $C$ is the smallest community containing $(1-o(1))$ fraction of $H$ and depth is $O(\log n)$,  each sub-community of $C$ can contain at most a constant fraction of $H$.  Then by claim~\ref{hcs-flat}, since $|C| \in \Omega(f(n))$ and so by assumption the number of inter-community edges in this decomposition is bounded by $o(f(n))$, $H$ is not an expander.     
 \end{proof}  
  

One of the main measures of ``quality'' of HCS is modularity value of the decompositions throughout the hierarchy: the higher values correspond to better structural properties. Let us start by relating modularity to the number of inter-community edges.  Note that while both the definition of expander and modularity are essentially relying on  estimating how far the number of edges across the "best" partition is from the expected value, the dependence on the size of the communities is somewhat different. 

First, to simplify our proof, consider decomposition into 2 communities; we can do it thanks to the following claim.  

\begin{claim}~\cite{newman2006modularity}
Suppose that $C_1 \dots C_t$ is a decomposition of $C$ optimizing modularity. Then for some $S=\{C_{i_1},\dots, C_{i_\ell}\}$ the partition into $S$ and $\overline{S}$ maximizes modularity among 2-partitions.     
\end{claim}

\newcommand{\eout}{\ensuremath{e_{out}}}
\begin{corollary} \label{2modularity}
For any graph G, the best modularity of a 2-partition is a lower bound for the modularity of an optimal partition, with edges in the the best 2-partition a subset of inter-community edges in the optimal partition. 
\end{corollary}

To simplify notation, for a set $S \subset V$, define $vol(S) = \Sigma_{v \in S} deg(v)$. Let $e_{out}(C)$ be the set of edges leaving a community $C_i$, so $\Sigma_{C \in P} e_{out}(C)$ is the set of all intercommunity edges in a partition $P$.    Now, we can restate the formula for modularity of a partition $P$  as follows, using the fact that $\Sigma_{C \in P} vol(C) = 2|E|$:  
\[ Q(P) = 1 - \frac{1}{2|E|} \Sigma_{C \in P} \eout(C)  - \frac{1}{4|E|^2}  \Sigma_{C \in P} (vol(C))^2   \] 

The optimal $Q$ of a graph $G$ is then $Q=Q(G) = \max_{P}  Q(P)$. Let $Q_2=Q_2(G) = \max_{P, |P|=2} Q(P)$ (that is, optimal modularity over 2-partitions).  Note that $Q(G) \geq Q_2(G)$ for any $G$. Let $P$ be a partition of $G$ into a set $S$ and its complement $\overline{S}$; we will always assume, without loss of generality, that $|S| \leq |\overline{S}|$. Let $|V|=n$ and $|E|=m$ (not to be confused with the number of clauses in the formula for which $G$ is a VIG).  Then  

\begin{align*}
Q(\{S,\overline{S}\}) & =  1 - \eout(S)/2m -\eout(\overline{S})/2m - (vol(S)/2m)^2 - (vol(\overline{S})/2|E|)^2 \\
& = 1 - \eout{S}/m - \frac{1}{4m^2} (vol(S)^2 + (2m-vol(S))^2) \\ 
& =   vol(S)/m -  vol(S)^2 / 2m^2 - \eout(S)/m \\
& = \frac{vol(S)}{m} (1 - vol(S)/2m)  - \eout(S)/m
\end{align*}

From there, we get that $\eout(S) =  vol(S)(1-vol(S)/2m) - Q(S)m$. Note that since $Q_2=max_S Q(\{S,\overline{S}\})$, $\eout(S) =  vol(S)(1-vol(S)/2m) - Q_2 m$.  Therefore, if $vol(S)(1-vol(S)/2m) - Q_2 m=o(f(n))$, then so is $\eout(S)$. 





\subsection{HCS Avoids the Resolution Limit in Some Graphs}
\label{sec:resolution_limit}

We prove Theorem~\ref{th:hcs-resolution}, restated below more formally.

{
\renewcommand{\thetheorem}{\ref{th:hcs-resolution}}
\begin{theorem}[restated]
  \label{th:hcs-resolution-formal}
  There exists a graph $G$ whose natural communities are of size
  $\log n$ and correspond to $\partition_h(G)$, while $\partition(G)$
  consists of communities of size $\sqrt{n/\log^3n}$.
\end{theorem}
\addtocounter{theorem}{-1}
}

We use as the separating graph $G$ the ring of cliques example
of~\cite{fortunato2007modularity}, which can be built as
follows. Start with a collection of $q=n/c$ cliques $C_1,\ldots,C_q$,
each of size $c=\log n$. Fix a canonical vertex $v_i$ for each clique
$C_i$. Add edges between $v_i$ and $v_{i+1}$, wrapping around at
$q$. The number of vertices is $n$ and the number of edges is
$m=\frac{n}{c}\binom{c}{2}+\frac{n}{c}=n(c+1)/2+o(n)$. The degree of
most vertices is $c-1$, except for the canonical vertices which have
degree $c+1$. The natural partition of $G$ into communities is the set
of cliques $C_1,\ldots,C_q$.

We say that a subgraph of $G$ preserves the cliques if for every
clique, it contains either all or none of its vertices. We say that a
partition of $G$ preserves the cliques if every element of the
partition preserves the cliques.

To prove that $\partition_h(G) = \set{C_1,\ldots,C_q}$ we need the
following two observations.

\begin{lemma}
  \label{lem:preserve-cliques}
  Let $H$ be a subgraph of $H$ that preserves the cliques. Then any
  partition of $H$ optimizing subgraph modularity preserves the
  cliques.
\end{lemma}

\begin{proof}
  Let $V_1,\ldots,V_k$ be a partition and assume that clique $C_i$
  contains vertices in at least two different sets $V_j$. We claim
  that the partition $V_1\setminus C_i,\ldots,V_k\setminus C_i,C_i$
  has greater modularity and that the number of split cliques
  decreases, therefore iterating this procedure until no clique is
  split concludes the lemma.

  To prove the claim, let $U_j = V_j \cap C_i$ and note that the change in modularity is at least
  \begin{align}
    2m\Delta Q
    &\geq - 4\left(1-\frac{(c+1)^2}{2m}\right) + \sum_{j \neq j'} \sum_{u\in U_j}\sum_{v\in U_{j'}}\left(1-\frac{d(u)d(v)}{2m}\right) \\
    &\geq -4 + \sum_{j \neq j'} \setsize{U_j}\setsize{U_{j'}}\left(1-\frac{(c+1)^2}{2m}\right) \\
    &\geq -4 + \frac{1}{2}\sum_{j \neq j'} \setsize{U_j}\setsize{U_{j'}} \\
    &\geq -4 + \frac{1}{4}\sum_{j \neq j',\setsize{U_j}>0,\setsize{U_j'}>0} \setsize{U_j}+\setsize{U_{j'}} \\
    &\geq -4 + c/4 > 0 \qedhere
  \end{align}
\end{proof}

\begin{lemma}
  \label{lem:separate-cliques}
  Let $H$ be a subgraph of $H$ that preserves the cliques. If $H$
  contains $q \geq 2$ cliques, then any partition of $H$ optimizing
  subgraph modularity contains at least $2$ elements.
\end{lemma}

\begin{proof}
  Let $V_1,V_2$ be the partition where the first half of the cliques
  in $H$ are contained in $V_1$ and the rest in $V_2$. The change in
  modularity with respect to the singleton partition is at least
  \begin{align}
    2m\Delta Q
    &\geq -2\left(1-\frac{(c+1)^2}{2m}\right) + \sum_{u\in V_1}\sum_{v\in V_2} \frac{d(u)d(v)}{2m} \\
    &\geq -2 + (cq/2)^2\frac{(c-1)^2}{2m}
    = -2 + \frac{n^2c^2}{9nc} > 0 \qedhere
  \end{align}
\end{proof}

We can now compute $\partition_h(G)$ by induction, using the induction
hypothesis that every node of the hierarchical tree is a subgraph that
preserves the cliques. The root is the whole graph $G$ and trivially
preserves the cliques. Every node of the hierarchical tree preserves
the cliques, then by Lemma~\ref{lem:preserve-cliques} all its children
preserve the cliques. This implies that all of the leaves preserve the
cliques. However, a leaf cannot consist of more than one clique,
otherwise it would contradict Lemma~\ref{lem:separate-cliques}. It
follows that all the leaves are single cliques as we wanted to show.

Next we prove that the partition optimizing modularity has large
elements, for which we need the following observations.

Let $V_1,\ldots,V_k$ be a partition that preserves the cliques. We
define the operation ``move clique $a$ to position $b$'' as
follows. Assume $a<b$. For $a \leq i < b$ we assign clique $i$ to the
set containing clique $i+1$, and clique $b$ to the set containing
clique $a$. Analogous if $a>b$.

\begin{lemma}
  \label{lem:contiguous}
  The modularity of $G$ is maximized by a partition of contiguous cliques.
\end{lemma}

\begin{proof}
  By Lemma~\ref{lem:preserve-cliques} the partition consists of unions
  of cliques. Assume a set contains non-contiguous cliques. Moving two
  intervals of cliques next to each other increases the modularity,
  therefore we can keep repeating this procedure until the sets only
  consist of contiguous intervals.
\end{proof}

In what follows we assume that $n$ is large enough for it not to make
it a difference when we treat variables as being continuous when they
are in fact discrete.

\begin{lemma}
  \label{lem:homogeneous}
  The modularity of $G$ is maximized by a partition of equal-sized
  elements.
\end{lemma}

\begin{proof}
  By Lemma~\ref{lem:contiguous} the partition consists of intervals of
  cliques. Assume interval $V_a$ is smaller than interval $V_b$. Let
  $C_i$ be an endpoint of $V_b$. Then assigning clique $C_i$ to $V_a$
  and moving it next to an endpoint of $V_a$ increases the modularity,
  therefore we can repeat this procedure until the sets are balanced.
\end{proof}

\begin{lemma}
\label{lem:largePartitions}
  The modularity of $G$ is maximized by a partition with elements of
  size $\sqrt{n/\log^3 n}$.
\end{lemma}

\begin{proof}
  By Lemma~\ref{lem:homogeneous} the partition consists of equal-sized
  intervals of cliques. Let $k$ be the size of the partition and let
  $q$ be the number of cliques in each block. The modularity is
  \begin{align}
    2mQ
    &= k \biggl(qc(c-1)+2(q-1) - \frac{1}{2m}\Bigl(q^2(c+1)^2+{}\\
    &\qquad{}+ q(q(c-1))(c+1)(c-1) + (q(c-1))^2(c-1)^2 \Bigr)\biggr) \\
    &= (n/cq) \left( q(c^2-c+2) - 2 - \frac{q^2}{2m}\left((c+1)^2+(c-1)^2(c+1)+(c-1)^4\right) \right) \\
    &= (n/c)(c^2-c+2) -2n/cq - \frac{nq}{2mc}\left((c+1)^2+(c-1)^2(c+1)+(c-1)^4\right)
  \end{align}
  which is a function of the form $a_0 - a_1q^{-1} - a_2q$
  and is maximized when its derivative is $0$ at point
  \begin{equation}
    q=\sqrt{\frac{a_1}{a_2}} \geq \sqrt{\frac{2n/c}{nc^3/2m}} \geq \sqrt{\frac{2n}{c^3}} \qedhere
  \end{equation}
\end{proof}

This completes the proof of Theorem~\ref{th:hcs-resolution-formal}.

\subsection{Justification for Parameter Choices and Lower Bounds for HCS}\label{sec:counterexamples}


For a formula to have ``good" HCS, we require that a number of parameters fall within appropriate ranges. One could hope that a single parameters of HCS might be sufficient in order to guarantee tractability, while also capturing a sufficiently large set of interesting instances. A natural candidate parameter would be \emph{high smooth modularity}: The modularity of the VIG is large, and the decrease in modularity from the parent to a non-leaf child in the HCS is sufficiently bounded. This captures the recursive intuition of HCS: each community should either be a leaf or should have a good partition into communities; this is closely related to the average modularity at each level, which is a parameter used in our experiments. However, in this section we show that if we are after provable tractability, it is unlikely that a single parameter of HCS will suffice. In doing so, we motivate using an ensemble of parameters as we have done in our experiments, as well as justify why we have chosen many of the parameters that we have.

Formally, we show that combinations of the following ranges of parameters admit formulas which are exponentially hard to refute in resolution.
\begin{itemize}
    \item \emph{High Root Modularity}: The modularity of $G$ is sufficiently large.
    \item \emph{High Smooth Modularity}: The modularity of $G$ is sufficiently large and the modularity of every non-leaf node is bounded below by a function of its level. 
    \item \emph{Bounded Leaf Size}: Each of the leaf-communities of the optimal HCS decomposition is sufficiently small --- in particular, of size $O(\log n)$.
    \item \emph{High Depth}: There is a sufficiently deep path in the optimal HCS decomposition.
    \item \emph{Minimum Depth}: There is a lower bound on the depth of every path in the optimal HCS decomposition.
    \item \emph{Dense Inter-Community Edges}: The optimal HCS decomposition has $\Omega(n^\varepsilon)$ inter-community edges between for all communities of size $O(n^\varepsilon)$
\end{itemize}
One finding that we would like to highlight is that in order to enure that the formula is tractable, the leaf-communities of the optimal HCS decomposition cannot be both large (of size $\omega(\log n)$) and unstructured; indeed, HCS says nothing about the structure of the leaf-communities. 

Finally, we note that we can still construct hard formulas which have ``good'' HCS parameters. However, the construction of these formulas is highly contrived and non-trivial and we do not see a way to simplify them. Indeed, they are far more contrived than the counterexamples to modularity given by \cite{mull-sat16}. We take this as empirical evidence that instances with good HCS avoid far more hard examples than formulas with high modularity.

Throughout, it will be convenient to make use of the highly sparse VIGs provided by random CNF formulas. For positive integer
parameters $m,n,k$, let $\mathcal{F}(m,n,k)$ be the uniform
distribution on formulas obtained by picking $m$ $k$-clauses on $n$
variables uniformly at random with
replacement. 
It is well known that for any $k$, the satisfiability of this
$F \sim \mathcal{F}(m,n,k)$ is controlled by the clause density
$\Delta_k:= m/n$: there is a threshold of $\Delta_k$ after which
$F \sim \mathcal{F}(m,n,k)$ becomes unsatisfiable with high
probability. Furthermore,
random $k$-CNF formulas near this threshold ($\Delta_k = O(2^k)$
suffices) require resolution refutations of size $2^{n^\varepsilon}$
for some constant $\varepsilon >0$ with high probability~\cite{chvatalS88,beameP96}.
The VIG of such a formula is sparse (it has $\Theta(n)$ edges) and with high probability
the maximum degree is $O(\log n)$. Furthermore, with high probability $F$ is expanding, and therefore the edges are distributed roughly uniformly throughout the VIG.

For brevity, our arguments section will be somewhat informal; however, it should be clear how to formalize them.

\paragraph{ \bf Root Modularity.} Mull, Fremont, and
Seshia \cite{mull-sat16} proved that having a highly modular VIG does not suffice to guarantee short resolution refutations. To do so, they extended the lower bounds on the size of resolution refutations of random $k$-CNF formulas \cite{chvatalS88,beameP96} to work for a distribution of formulas
whose VIGs have high modularity, thus showing the existence of a large
family of hard formulas with this property.
    
A much simpler proof of their result\,---\,albeit with slightly worse
parameters\,---\,can be obtained as follows: Let
$F \sim \mathcal{F}(m,n,k)$ with $k,m$ set appropriately so that $F$
is hard to refute in resolution with high probability, and let $F'$ be
the formula obtained by taking $t$ copies of $F$ on disjoint sets of
variables. It can be checked that the modularity of the partition of the VIG of $F'$ which has $t$ communities, one corresponding to each of the copies of $F$ has modularity $1-o(1/t)$.
Setting $t$ sufficiently large, we obtain a formula whose VIG has high modularity. As each copy of $F$ is on distinct variables,
refuting $F'$ is at least as hard as refuting $F$. Thus, a lower bound
of $2^{\Omega(n^\varepsilon)}$ for some constant $\varepsilon>0$
follows from the known lower bounds on refuting $F$ in resolution. If
we let $v =nt$ be the number of variables of $F'$ then this lower
bound is of the form $2^{\Omega((v/t)^\varepsilon)}$ and is 
superpolynomial provided $t =o(n/\log n)$. This argument also applies
to hierarchical community structure.

Observe that each leaf-community of $F'$ is a formula $F$ on $v/t$
variables which is hard to refute in resolution. This shows that if we
want to ensure polynomial-size resolution proofs, then we cannot allow the leaf
communities of the HCS decomposition to be both unstructured and large
(of size $\omega(\log n)$). The simplest way to avoid this is to
require the size of the leaf-community to be bounded by $O(\log n)$. However, in a later paragraphs we will show that this restriction is not sufficient on its own.

\paragraph{\bf Root Modularity and Maximum Depth} The previous example can be modified in order to rule out requiring an \emph{upper bound on the depth} and  \emph{high smooth modularity}. 
We will use the simple observation the at the optimal community structure decomposition of $t$ disjoint cliques on $n$ vertices has $t$ communities, one for each clique. Furthermore, HCS will not decompose any of these cliques, and therefore the optimal HCS has a single level and maximum depth assumption. 
 
Let $F'$ be the formula constructed in the previous example. We will modify each copy of $F$ so that its VIG is a clique $K_n$. Let $i \in [t]$, and for the $i$th copy of $F$ do the following: pick a variable $\hat{x}_i \neq v_i^*$ and note that there is a direction $\alpha_i \in \{0,1\}$ in which it can be set such that the complexity of refuting $F\restriction( \hat{x} = \alpha_i)$ in resolution is at most half the complexity of refuting $F$. Add an arbitrary clause of length $n$ containing every variable in $F$ such that $\hat{x}$ occurs positively if $\alpha_i=1$ and negatively otherwise. Observe that the VIG of $F$ is now a clique. 

After this process, the VIG of $F'$ has $t$ disjoint cliques. It remains to see that $F'$ is still hard to refute in resolution. This follows because applying the restriction which sets $\hat{x}_i = \alpha_i$ for all $i \in [t]$ leaves us with $t$ disjoint copies of $F'$ which require $2^{\Omega(n^{\varepsilon})}$ size resolution refutations, and because resolution complexity is closed under restriction.
    
\paragraph{\bf Smooth and Bounded Leaf Size.} Next, we show that
requiring the optimal HCS decomposition to have \emph{high smooth modularity} and to have 
\emph{bounded leaf size} is not sufficient in order to ensure small
resolution refutations.
Let $F \sim \mathcal{F}(m,n,3)$ be a random
$3$-CNF formula with $m$ set so that $F$ is hard to
refute in resolution whp. Let $G$ be the VIG of $F$, which has $\Theta(n)$ edges and degree $O(\log n)$ whp. Let $p = O(\log n)$ and $K_p$ be a clique
on $p$ vertices. Let $2 \leq t \leq p$ be any integer and let $F_K$ be
any $t$-CNF formula on $p$ variables such that every pair of variables occurs in some clause; this requires at most $p^2$ clauses.  Observe that the VIG
of $F_K$ is $K_p$. Furthermore, observe that $F_K$ is satisfied by any truth
assignment that sets at least $p-t+1$ variables to true.
    
Using $K_p$ and $G$ we construct a family of formulas that is hard to
refute in Resolution and whose VIG has the desired properties. Let $G'$ be the rooted product of $G$ and $K_p$,
that is let $v^*$ be an arbitrary vertex of $K_p$, create $n$ copies of
$K_p$, and identify the $i$th vertex of $G$ with
the vertex $v^*$ of the $i$th copy of $K_p$.

Next, we show that because $G$ is sparse and each $K_p$ is a clique, $G'$ has modularity which tends to $1$ with $n$. 

\begin{lemma}
  \label{lem:sparse-clique-modularity}
  Let $G$ be a graph of order $n$, size $m=\Theta(n)$, and degree
  $O(\sqrt{n})$; let $K_p$ be a clique on $p$ vertices with $p = \omega(1)$ and $p = o(n)$;
  and let $G'$ be the rooted product of $G$ and $K_p$. Then
  $Q(G') = 1 - o(1)$.
\end{lemma}

\begin{proof}
  Let $m' \leq m + np^2/2$ be the number of edges of $G'$.
  Consider the partition $P$ given by the $n$ copies of $K_p$. We have
  \begin{equation}
    Q(P) \geq \frac{np^2}{2m'}\left[1-\frac{(d+p)^2}{2m'}\right] \geq \left[1-\frac{m}{m'}\right]\left[1-\frac{(d+p)^2}{2m'}\right]
  \end{equation}
  and $m/m'=o(1)$ because $p = \omega(1)$, while $(d+p)^2/2m' = o(1)$ because
  $d=O(\sqrt{n})$ and $p=o(n)$. Hence $Q(G') \geq Q(P) = 1-o(1)$.
\end{proof}

As well, by a similar argument to the proof of Lemma~\ref{lem:preserve-cliques}, the optimal HCS decomposition will preserve cliques at each level (i.e. none of the cliques $K_p$ will be decomposed during the HCS decomposition. Furthermore, by an argument similar to Lemma~\ref{lem:largePartitions}, each node in the HCS decomposition will have too many children. Finally, observe that because the variables of the random CNF formula $F$ are distributed uniformly, the non-clique edges in $G$ will be distributed approximately uniformly between the cliques. Taken together, this implies that the modularity of every non-leaf node in the HCS cannot decrease too much compared to the modularity of its parent. In particular, we satisfy the high smooth modularity condition.

It remains to show that we can construct a formula with the same VIG which is hard to refute in resolution. Let $F'$ be the formula obtained in the same way: begin with the
formula $F$. Make $n$ copies of $F_K$ and choose some vertex $v_i^*$ from the $i$th copy of $F_k$. Identify $v_i^*$ with $x_i$, the $i$th variable of $F$. Observe that the VIG of $F'$ is exactly $G'$. It remains to argue that $F'$ is hard to refute in
Resolution. Let $V_i$ be set of variables on which the $i$th copy of
$F_K$ depends. Let $\rho \in \{0,1,*\}^{pn}$ be the restriction which
sets all variables in $V_i \setminus v^*_i$ to true for all
$i \in [n]$ and does not set each $v^*_i$ (that is $\rho(v^*_i) = *$
for all $v^*_i$). We claim that $F' \restriction \rho = F$. Indeed,
each clause in each $F_K$ depends on at least two variables, of which
at least one is set to true because $t \geq 2$. The only variables on
which $F$ depends are $v_i^*$ for $i \in [n]$, which have been left
untouched; thus $F \restriction \rho = F$. The lower bound on $F'$ follows from the $2^{\Omega(n^{\varepsilon})}$ lower bound on $F$ together with the fact that the complexity of resolution proofs is closed
under restriction.
    

\paragraph{\bf Deep, Smooth, Bounded Leaf Size.} We extend
the previous example to show that if, in addition to the previous
restrictions, we require the HCS decomposition have \emph{high
  depth}, then this is still insufficient to guarantee short resolution proofs.
  To see this, let $F'$ be the formula constructed in the
previous example, and let $F''$ be a formula containing only positive
literals and whose VIG satisfies smooth HCS, bounded leaf-community
size, and whose optimal HCS decomposition has a sufficiently deep
root-to-leaf path. Note that such a formula $F''$ exists if and
only if there exists a formula satisfying these properties with no
restriction on the polarity of the literals. Therefore, we can
assume the existence of $F''$, as otherwise this set of parameters
does not capture any instances and is therefore not useful. 

We claim that the VIG corresponding to $F' \wedge F''$ will satisfy their shared properties. First, observe that $F' \wedge F''$ has high modularity:
the modularity of both $F'$ and $F''$ is high by assumption and therefore the partition which is the union of the maximum modularity partitions of $F'$ and $F''$ will have high modularity. Furthermore, because $F'$ and $F''$ are disjoint, the HCS decomposition will never never put (a part of) both of these formulas in the same community. Thus, the maximum modularity partition is the union of the maximum partitions of $F'$ and $F''$. From then onwards, because $F'$ and $F''$ are smooth, the HCS decomposition of $F' \wedge F''$ will be smooth. As well, because $F''$ has a deep decomposition, this implies that the decomposition of $F' \wedge F''$ will also contain a deep path. 
Finally, because $F''$
is satisfiable and $F'$ requires exponential length refutations in
resolution, $F' \wedge F''$ does as well. 

\paragraph{\bf Lower Bounds when the Number of Inter-community Edges is Large.} Finally we show that, under the assumption that each large community has a non-negligible fraction of inter-community edges, we can construct a hard formula satisfying most ``good'' HCS parameters. We note that the requirement on the inter-community edges is fairly weak and we believe that it can be weakened even more.

Consider a formula $F$ with a VIG $G$ such that its hierarchical decomposition has $\Omega(n^\varepsilon)$ inter-community edges for all communities of size $O(n^\varepsilon)$; suppose also that the community degree of $F$ is a constant $c$.  We show that it is then possible to embed a graph of a sparse hard formula  of size $O(n^\varepsilon)$ into the graph of $F$, and construct a formula with the same VIG so that refuting this new formula amounts to refuting the sparse hard formula.  

Let $F_H$ be a (family of) hard formulas on $n'=n^\varepsilon$ variables with a constant clause width $k$ such as each variable occurs in constantly many clauses $\Delta$; for example, $F_H$ could be a Tseitin formula on a 3-regular expander. We choose $F_H$ to be such that VIG of $F_H$ is a $d$-regular expander, where $d=\Delta\cdot w$; this is satisfied by formulas resulting from the Tseitin contradictions on expanders.  

Let $F$  be a (family of) easy formulas with a VIG $G_F$ with ``good'' HCS parameters except that its hierarchical decomposition has $\Omega(n^\varepsilon) \geq b n^\varepsilon$ inter-community edges for all communities of size $O(n^\varepsilon)$, and such that the degree of vertices in its VIG is at least $D$ for some large enough $D$.   Note that the number of communities of this size is at least $2n^{1-\varepsilon}$, so the average number of inter-community edges out of each vertex of $G_F$ is at least $2b = 2 n^{1-\varepsilon} \cdot b n^\varepsilon / n$. To simplify calculations, take $D \geq 4d + 4b$ (note that $D$ can still be a constant, just a larger constant than $d$).

Now, add $n^\varepsilon$ new vertices to $G$, placing each of them in a separate leaf of the hierarchical decomposition of $G_F$. Associate each of these vertices with a variable of $F_H$, and add edges corresponding to clauses of $F_H$; call the resulting subgraph $G_H$. Now, for each leaf containing a new variable $v^*_i$, add $v^*_i$ to enough clauses within this leaf to bring the degree of $v^*_i$ up to $D$ (here, assume that these leaves are of size at least $D$).   Finally, create a new formula $F'=F^+ \wedge F_H$, where $F^+$ is $F$ with all negations removed (that is, every literal of $F^+$ is positive and so $F^+$ is trivially satisfiable). Therefore, to refute $F'$ it is necessary to refute $F_H$. If $F_H$ is a Tseitin contradiction on a 3-regular expander, then the size of its resolution proof is $2^{\Omega(n')}=2^{\Omega(n^\varepsilon)}$ \cite{ben2001short,urquhart1987hard}. 

\begin{claim}\label{hidden-hard} 
 HCS properties of the VIG of $F'$, $G_{F'}$, are indistinguishable from those of the VIG of $F$.
\end{claim}
\begin{proof}[Proof outline] 
First, consider the root of the hierarchical decomposition, i.e., the first decomposition of $G_{F'}$. We can show that the optimal partition of $G_{F'}$ is very close to the optimal partition of $G_F$ because new vertices form a very small part of every community, and every vertex of $G_H$  has many more edges to the vertices outside of $G_H$ in $G_{F'}$ than to vertices within $G_H$. Let community degree be $c$ and suppose the expansion of $G_H$ is $\alpha$; then, in the worst case, the decomposition will contribute $c(c-1)\alpha n^\varepsilon / 2c = \frac{c-1}{2} \alpha n^\varepsilon$ edges. But as we assumed that the number of edges is already $\Omega(n^\varepsilon)$, this will add at most a constant factor to the total number of inter-community edges at the root level. Now we can repeat this argument for the subcommunities, which, since vertices were selected randomly, with high probability will contain a number of variables of $G_H$ proportional to their size, and have a diminishing fraction of edges of $G_H$.  
\end{proof}

}
\end{document}